\documentclass[12pt,onecolumn]{IEEEtran}
\usepackage{graphicx,color,epsfig,rotating}
\usepackage{epsf}
\usepackage{latexsym,url}
\usepackage{subfigure}
\usepackage{cite}
\usepackage{amsmath}
\usepackage{amssymb}

\long\def\comment#1{}

\newfont{\bb}{msbm10 scaled 1000}
\newcommand{\CC}{\mbox{\bb C}}

\newcommand{\EE}{\mbox{\bb E}}

\newfont{\bbsmall}{msbm10 scaled 700}

\newcommand{\ev}{{\bf e}}

\newcommand{\wv}{{\bf w}}

\newcommand{\xv}{{\bf x}}
\newcommand{\yv}{{\bf y}}

\newcommand{\Gm}{{\bf G}}
\newcommand{\Hm}{{\bf H}}
\newcommand{\Id}{{\bf I}}

\newcommand{\Rm}{{\bf R}}

\newcommand{\Wm}{{\bf W}}

\newcommand{\Xm}{{\bf X}}
\newcommand{\Ym}{{\bf Y}}

\newcommand{\Cc}{{\cal C}}
\newcommand{\Dc}{{\cal D}}

\newcommand{\Kc}{{\cal K}}

\newcommand{\Nc}{{\cal N}}
\newcommand{\Oc}{{\cal O}}

\newcommand{\alphav}{\hbox{\boldmath$\alpha$}}

\renewcommand{\det}{{\hbox{det}}}
\newcommand{\trace}{{\hbox{tr}}}

\newcommand{\herm}{{\sf H}}

\newtheorem{thm}{Theorem}

\newtheorem{lem}[thm]{Lemma}

\newtheorem{note}{Remark}

\newcommand{\mmse}{{\sf mmse}}

\title{Channel State Feedback over the MIMO-MAC}

\author{K.~Raj~Kumar and
        Giuseppe Caire,~\IEEEmembership{Fellow,~IEEE}%
\thanks{K. R. Kumar is with the Laboratoire d'algorithmique, Ecole Polytechnique F\'{e}d\'{e}rale de Lausanne,
1015 Lausanne, Switzerland ({\tt raj.kumar@epfl.ch}). G.~Caire is with the Department of Electrical Engineering - Systems,
University of Southern California, Los Angeles, CA 90089, USA ({\tt
caire@usc.edu}). Much of this work was carried
out while K.~R.~Kumar was with the University of Southern California.}%
\thanks{The material in this paper was presented in part at the IEEE International Symposium on Information
Theory (ISIT-09), Seoul, Korea, Jun.~28 - Jul.~3, 2009.}%
}

\begin{document}

\maketitle

\markboth{Submitted to IEEE Transactions on Information Theory}%
{Kumar and Caire: Channel State Feedback over the MIMO-MAC}

\maketitle

\begin{abstract}
We consider the problem of designing low latency and low complexity
schemes for channel state feedback over the MIMO-MAC (multiple-input multiple-output
multiple access channel). We develop a framework for analyzing this
problem in terms of minimizing the MSE distortion, and come up with
separated source-channel schemes and joint source-channel schemes
that perform better than analog feedback.  
We also develop a strikingly simple code design based on scalar quantization and uncoded QAM modulation that 
achieves the theoretical asymptotic performance limit of the separated approach with very low complexity and latency, in the case of
single-antenna users. 
\end{abstract}



\section{Introduction}

\IEEEPARstart{C}{onsider} a frequency division duplex (FDD) cellular system with
sufficient frequency spacing between the uplink and downlink
channels, such that the uplink and downlink fading coefficients are
independent.  A base station (BS) with $M$ antennas wishes to serve
$K$ user terminals (UTs), with $N_t$ antennas each, using some MIMO
broadcast channel precoding technique, such as linear beamforming,
Dirty-Paper Coding or some low-complexity non-linear precoding
approximation thereof.  Essential to these techniques is the
availability of accurate channel state information  at the
transmitter (CSIT), that is, the BS must know the user downlink
channels. 
We consider the popular block-fading model, according to which the channel coefficients remain constant for
time-frequency slots of some finite but large number of channel uses (complex dimensions), 
and change independently from slot to slot.  We  assume that the UTs have perfect knowledge of their
downlink channels on each block, which can be obtained from 
downlink training symbols broadcasted by the BS.  
Then, at each time slot, the downlink channel coefficients must be 
fed back to the BS on the uplink. We can model this CSIT feedback as signaling 
over a MIMO multiple-access channel (MAC). It has been shown in
\cite{Jin,CaiJinKobRav} that the relevant performance measure that
dominates the downlink rate gap\footnote{The gap between the rate
achievable with imperfect CSIT from the optimal rate achievable with
perfect CSIT and the optimal MIMO broadcast channel coding
strategy.} is the mean-square error (MSE) distortion 
at which the BS is able to represent the UTs channel coefficients (this will be made more
precise in the sequel). It follows that the CSIT feedback problem
consists of lossy transmission under an end-to-end 
MSE distortion constraint of a set of $K$ Gaussian sources 
(i.e., the UTs downlink channel coefficients) over a MIMO-MAC channel 
affected by block fading. 
It should also be noticed that, apart from some ``analog''
feedback schemes where the channel coefficients are 
sent unquantized and in parallel by all users (see \cite{Jin,CaiJinKobRav} and
references therein), very little attention has been devoted to
properly design the CSIT feedback scheme exploiting the MIMO-MAC
nature of the uplink channel: most works assume perfect feedback at
fixed rate, or implicitly assume that the feedback information is
piggybacked ``somehow'' into the uplink transmissions. This may pose
problems, since the CSIT feedback must have extremely low latency,
therefore, its coding block length must be very short. In this
paper we consider the problem of designing very low latency and
low complexity CSIT feedback  schemes for the uplink MIMO-MAC.

\section{Channel Model and Problem Statement\label{sec:channel_model}}

Although the fading slots may span a large number of channel uses, we focus here only on the dimensions effectively used by the 
CSIT feedback scheme. We assume that the CSIT feedback uplink transmission spans 
a block of $T$ channel uses of the complex baseband equivalent MIMO-MAC channel, represented by 
\begin{equation} \label{eq:MAC}
\Ym = \sqrt{\rho} \sum_{i=1}^K \Hm_i \Xm_i + \Wm,
\end{equation}
where $\Ym$ denotes the $M \times T$ signal received at the BS,
$\Xm_i$ is the $N_t \times T$ signal transmitted by the
$i^{\text{th}}$ UT, $\Hm_i$ is the (uplink) channel between the
$i^{\text{th}}$ UT and the BS that is assumed to remain constant
over $T$ channel uses, and $\Wm$ is the additive noise. We assume
that the entries of $\Hm_i$ and $\Wm$ are distributed as i.i.d.
complex Gaussian with zero mean and unit variance $\mathcal{CN}
(0,1)$.  As anticipated before, we assume that the channels remain constant for a
such block of $T$ channel uses. 
We impose an average per-user power constraint
\[ \mathbb{E} \left[ \left\| \Xm_i \right\|_F^2 \right] \leq T \ \ \forall \ i, \]
where $\| \cdot \|_F$ denotes the Frobenius norm. Hence $\rho$ takes
on the meaning of the {\em uplink} transmit SNR.

We consider that each UT needs to transmit $S$ samples of a complex
i.i.d. Gaussian source\footnote{The source length $S$ corresponds to
the overall number of channel coefficients that need to be fed back
at any time slot.}  over $T$ uses of the MIMO-MAC. We define $b
\triangleq T/S$ as the {\em bandwidth efficiency} of the feedback scheme, measured as the number 
$b$ of feedback (uplink) channel uses per source sample (i.e., per downlink channel coefficient). 
Therefore, it is clear that the bandwidth efficiency immediately relates to the ``protocol overhead'' incurred by a closed-loop 
channel state feedback scheme. Associated with a particular $b$, we consider a family of {\em coding schemes}
$\mathcal{C}(\rho)$ indexed by the SNR $\rho$. Corresponding to the
coding scheme $\mathcal{C}(\rho)$, we define $D(\rho)$ to be the MSE distortion involved in reproducing the
source, averaged over the source, channel and the noise. The {\em
distortion SNR exponent} of the family is given by \cite{CaiNar,GunErk,BhaNarCai}
\[ \delta (b) = - \lim\limits_{\rho \rightarrow \infty} \frac{\log D(\rho)}{\log \rho}.\]
The distortion SNR exponent {\em of the channel} $\delta^*(b)$ is
given by the supremum of $\delta(b)$ over all possible coding
schemes.

The operational significance of the distortion exponent defined above for the CSIT feedback problem is
provided next, as the main motivation for the results of this paper.
The key quantity that dominates the rate gap in the MIMO downlink
with imperfect CSIT is the system self-interference, due to imperfect 
cancellation of the multiuser interference due to the imperfect knowledge of the downlink channels.
The self-interference average power is given by the product between the 
downlink transmitted signal power and the channel MSE distortion at the BS. 
Up to constants, this is given by \cite{Jin,CaiJinKobRav} $\rho \cdot D(\rho)$, 
provided that the uplink and the downlink SNRs differ by a constant factor.
In the high-SNR limit ($\rho \rightarrow \infty$), 
if the above product vanishes, then
the asymptotic high-SNR rate gap goes to zero and the imperfect CSIT
scheme yields close-to-optimal performance. If the above product
converges to a constant, then a constant rate gap is achieved. If
the above product grows as a power $\rho^\mu$, with  $0 \leq \mu \leq 1$, then the downlink
sum-rate scales only as $\min\{M,KN_t\} (1 - \mu) \log \rho$, i.e.,
a loss of degrees of freedom by the factor $1 - \mu$ occurs.
Finally, if the above product grows as $\rho$ (i.e., $D(\rho) = O(1)$), then the user rates saturate to
a constant term and the system becomes self-interference limited (unless
only a single user is served at each slot, of course). 
It is therefore clear that the MIMO downlink performance depends
critically on the distortion exponent $\delta(b)$. 

From a practical viewpoint it should be noticed that the dynamic range of SNR in a cellular system in the presence of distance-dependent pathloss
may be of the order of 40 dB, i.e., the receiver SNR of UTs near the BS may be 40 dB larger than the received SNR of UTs at the cell edge.
Hence, it is essential to design a CSIT feedback scheme whose MSE distortion $D(\rho)$ 
{\em improves} with the received SNR $\rho$. 
Schemes based on quantization codebooks of fixed size, provide $D(\rho) = O(1)$ (i.e., $\delta(b) = 0$ for any $b$). 
The analog feedback scheme (see \cite{Jin,CaiJinKobRav}  and references therein) 
achieves $\delta(b) = 1$ for all $b \geq 1$. This is optimal for $b = 1$, but generally suboptimal
for $b  > 1$.  Our goal is to study more efficient techniques that achieve $\delta(b) > 1$ at some $b \geq 1$.

We first present an upper bound on the distortion SNR exponent in
Sec.~\ref{sec:upper bound} and then analyze the performance of a
separated source-channel coding scheme in Sec.~\ref{sec:separated}.
The sub-optimality of the separated scheme prompts us to investigate
joint source-channel coding schemes for the MIMO-MAC, following
which we present a hybrid digital-analog coding scheme in
Sec.~\ref{sec:hybrid}. In Sec.~\ref{sec:QAM_DMT}, we construct simple channel codes that have optimal 
performance for the important special case where the UTs have a single antenna each. Simulation results that highlight the
superiority of digital feedback over conventional analog feedback,
and establish the relevance of the distortion SNR exponent in
practical scenarios are presented in Sec.~\ref{sec:simul}.

\section{Upper bound on the distortion SNR exponent\label{sec:upper bound}}

Our upper bound is inspired by the one presented for the single-user
MIMO block fading channel in \cite{CaiNar,GunErk}.
First, notice that we consider the simple case of UTs with symmetric
statistics, and since the UTs have no uplink channel state
information, they all send the same coding rate, by symmetry. Define
$\Hm = [\Hm_1 \Hm_2 \cdots \Hm_K]$. Consider an augmented channel
where all UTs are provided with the knowledge of the coding rate
$R(\Hm)$ that can be transmitted such that the rate point $\Rm(\Hm)
= (R(\Hm),\ldots,R(\Hm))$ is inside the capacity region of the
MIMO-MAC with given channel vectors and transmit SNR $\rho$, denoted
by $\Cc_{\text{MAC}}(\Hm; \rho)$. Using this knowledge, each UT can
employ a separated source-channel coding scheme with a source coding
rate $R_s = b R(\Hm)$ nats per complex sample, and achieve an
end-to-end instantaneous distortion given by $\Dc(\Hm) = \exp(-b
R(\Hm))$. This would result in the average distortion $D(\rho) = \EE[
\exp(-bR(\Hm))]$.

From the expression of the MIMO-MAC capacity region $\Cc(\Hm;\rho)$
we have the following conditions on the rate $R(\Hm)$,
\begin{equation} \label{rate-bound-difficult}
R(\Hm) \leq \frac{1}{|\Kc|} \log \det \left( \Id + \rho \sum_{k \in
\Kc} \Hm_k \Hm_k^\herm \right),
\end{equation}
for all subsets $\Kc \subseteq \{1,\ldots,K\}$, where $|\Kc|$
denotes the cardinality of the set $\Kc$. Denote $\Hm_\Kc$ to be the
equivalent $M \times |\Kc|N_t$ channel comprising of the channels of
the users in $\Kc$ stacked next to each other. We can rewrite
\eqref{rate-bound-difficult} as
\begin{equation*}
R(\Hm) \leq \frac{1}{|\Kc|} \log \det \left( \Id + \rho \Hm_\Kc
\Hm_\Kc^\herm \right).
\end{equation*}
Let $m_\Kc \triangleq \min\{ M,|\Kc|N_t \}$ and $n_\Kc \triangleq
\max\{ M,|\Kc|N_t \}$. We define $\lambda_1 \leq \cdots \leq
\lambda_{m_\Kc}$ to be the $m_\Kc$ ordered non-zero eigenvalues of
$\Hm_\Kc \Hm_\Kc^\herm$, and rewrite the above as
\begin{equation*}
R(\Hm) \leq \frac{1}{|\Kc|} \log \left[ \prod_{i=1}^{m_\Kc} (1 +
\rho \lambda_i) \right].
\end{equation*}

Shannon's source-channel separation theorem holds for a single-user
system and any fixed channel. Therefore, as far as each of the above
constraints are concerned, we have a situation completely equivalent
to a single-user MIMO channel where the $|\Kc|$ users act as a
single cooperative transmitter with a source of length $|\Kc|S$,
operating over a MIMO channel with block length $T$ and channel
coding rate $|\Kc|R(\Hm)$. It follows that the corresponding
distortion obtained by this genie-aided scheme is a lower bound on
any achievable distortion for the actual MIMO-MAC. This is given by
\begin{eqnarray} \label{eq:D_LB}
D^{\rm LB}(\rho) & = & \mathbb{E} \left[ \exp(-b R(\Hm)) \right] \nonumber\\
&\geq& \mathbb{E} \left[ \prod_{i=1}^{m_\Kc} \left( 1 + \rho
\lambda_i \right)^{-\frac{b}{|\Kc|}} \right].
\end{eqnarray}
We set\footnote{We use the exponential equality notation of
\cite{ZheTse},
\[ x \doteq y \Leftrightarrow \lim\limits_{\rho \rightarrow \infty} \frac{\log x}{\log \rho} = \lim\limits_{\rho \rightarrow \infty} \frac{\log y}{\log \rho}\]
} $\lambda_i \doteq \rho^{-\alpha_i}$. The joint pdf of the random
vector $\alphav = (\alpha_1,\hdots,\alpha_{m_\Kc})$ is given as
\cite{ZheTse}
\begin{equation} \label{eq:joint_pdf}
p(\alphav) \doteq \left\{
\begin{array}{cc}
\rho^{-\sum_{i=1}^{m_\Kc} (2i-1+n_\Kc-m_\Kc) \alpha_i}, & \alpha_1
\geq \cdots \geq
\alpha_{m_\Kc} \geq 0\\
\rho^{-\infty}, & \text{ otherwise}
\end{array}
\right. .
\end{equation}
Using \eqref{eq:joint_pdf} to compute \eqref{eq:D_LB} and an
application of Varadhan's lemma as in \cite{ZheTse} results in an
upper-bound on the distortion SNR exponent
\begin{eqnarray*}
\delta^*(b) \leq \inf\limits_{\alpha_1 \geq \cdots \geq
\alpha_{m_\Kc} \geq 0} \sum_{i=1}^{m_\Kc} \left[ \frac{b}{|\Kc|}
(1-\alpha_i)^+ \right. \\
\left. \phantom{\frac{b}{|\Kc|}} + (2i-1+n_\Kc-m_\Kc) \alpha_i \right],
\end{eqnarray*}
where $(x)^+ \triangleq \max \{0,x\}$.
It can be verified that the
above infimum is achieved by
\[ \alpha_i^* = \left\{
\begin{array}{cc}
0, & \frac{b}{|\Kc|} \leq 2i-1+n_\Kc-m_\Kc\\
1, & \frac{b}{|\Kc|} > 2i-1+n_\Kc-m_\Kc
\end{array}
\right. .\] Hence
\[ \delta^*(b) \leq \sum_{i=1}^{m_\Kc} \min \left\{ \frac{b}{|\Kc|}, 2i-1+n_\Kc-m_\Kc \right\}, \]
$\forall \ \Kc \subseteq \{1,\hdots,K\}$. For a fixed $|\Kc|$,
define
\begin{equation} \label{eq:delta}
\delta_{|\Kc|} (b) = \sum_{i=1}^{m_\Kc} \min \left\{
\frac{b}{|\Kc|}, 2i-1+n_\Kc-m_\Kc \right\}.
\end{equation} We hence have
that
\[ \delta^*(b) \leq \min \{ \delta_1(b), \delta_2(b), \hdots, \delta_K(b) \}.\]
We term the above upper-bound as the {\em informed transmitter
bound}.

We now investigate the informed transmitter bound for the important
special case of $N_t = 1$ and $M = K$. Notice that for $b \leq K$,
we have that $\delta_1(b) = \delta_K(b) = b$. Let us examine how
the $\delta_{|\Kc|}$ in \eqref{eq:delta} behave, for $1 < |\Kc| < K$
in the range of $b \leq K$. In particular, we claim that
\[ \min \left\{
\frac{b}{|\Kc|}, 2i-1+K-|\Kc| \right\} = \frac{b}{|\Kc|}, \ i =
1,2,\hdots,|\Kc|.\] This is true for $b \leq K$ since for some
$\epsilon \geq 0$,
\begin{eqnarray*}
2i-1+K-|\Kc| < \frac{K}{|\Kc|} - \epsilon\\
\Leftrightarrow K < \frac{|\Kc|}{|\Kc|-1} (|\Kc|-2i+1-\epsilon)\\
\Leftrightarrow K < |\Kc| - |\Kc| \frac{2i-2+\epsilon}{|\Kc|-1},
\end{eqnarray*}
for $i = 1,2,\hdots,|\Kc|$. Setting $i=1$, this requires that
\begin{eqnarray*}
K < |\Kc| - \epsilon \frac{|\Kc|}{|\Kc|-1}\\
\Rightarrow K < |\Kc|,
\end{eqnarray*}
which is not true. This implies that $\delta_{|\Kc|}(b)=b \ \forall
\ b \leq K$ and $\forall \ 1 \leq |\Kc| \leq K$. In particular,
$\delta_{|\Kc|}(K)=K$. Now consider the case $b> K$. Notice that
$\delta_1(b) = K \ \forall \ b>K$. Since $\delta_{|\Kc|}(b)$ is an
increasing function of $b$, we have that
\begin{eqnarray*}
\delta_{|\Kc|}(b) &\geq& \delta_{|\Kc|}(K) \ \forall \ b \geq K\\
&=& K.
\end{eqnarray*}
This implies that
\[ \delta_{|\Kc|}(b) \geq \delta_1(b) \ \forall \ b. \]
The informed transmitter bound hence reduces to
\[ \delta^*(b) \leq \delta_1(b) = \min \{b,K \}.\]

\section{Separated source-channel coding for the MIMO-MAC\label{sec:separated}}

In this section, we compute the achievable distortion SNR exponent
under separated source-channel coding.

\subsection{Diversity multiplexing tradeoff (DMT) of the MIMO-MAC}
We first present a very brief overview of the DMT \cite{ZheTse} of
the MIMO-MAC \cite{TseVisZhe}, which is a metric of performance for
channel coding. According to the DMT formulation, we scale the rate
of transmission of the coding scheme $\mathcal{C}(\rho)$ as $r \log
\rho$, where $r$ denotes the {\em multiplexing gain}. We define the
{\em diversity gain} of the system to be
\[ d(r) = - \lim\limits_{\rho \rightarrow \infty} \frac{\log P_e(\rho)}{\log \rho}, \]
where $P_e(\rho)$ denotes the codeword error probability. The DMT
$d^*(r)$ is defined to be the supremum of all achievable diversity
gains. We consider a common diversity requirement of $d(r)$ for the
transmission of all UTs, and equal rate transmission. For the
general $K$ user MAC with $N_t$ antennas at each UT and $M$ antennas
at the BS, the DMT is given by \cite{TseVisZhe}
\[ d^*_{MAC}(r) = \left\{
\begin{array}{ll}
d^*_{N_t,M} (r), & r \leq \min \left\{N_t, \frac{M}{K+1} \right\}\\
d^*_{KN_t,M} (Kr), & r \geq \min \left\{N_t, \frac{M}{K+1} \right\}
\end{array} \right. , \]
where $d^*_{n_t,n_r}(r)$ is the DMT for the single user MIMO channel
with $n_t$ transmit and $n_r$ receive antennas \cite{ZheTse}, given
by the piecewise linear function interpolating the points
$(r,(n_t-r)(n_r-r))$ for integral $r = 0,1,\hdots,\min\{n_t,n_r\}$. 
The tradeoff performance can hence be divided into two regimes, the
lightly loaded regime corresponding to $r \leq \min \left\{N_t,
\frac{M}{K+1} \right\}$, and the heavily loaded regime corresponding
to $r \geq \min \left\{N_t, \frac{M}{K+1} \right\}$. In the lightly
loaded regime, the DMT of the MAC is as though there was only one
user in the system, i.e., single-user performance is achieved. In
the heavily loaded regime, the DMT of the MAC is as though all the
users pooled their antennas together into a single ``super-user'',
and transmit at $K$ times the single-user rate. The authors in
\cite{TseVisZhe} were also able to show that the dominant error
event in the lightly loaded regime corresponds to a single user
decoding in error, while the dominant error event in the heavily
loaded regime is that of all users decoding in error.

\subsection{Distortion exponent with separated source-channel coding}
We now turn to the analysis of the distortion SNR exponent of a
separated source-channel coding scheme for the MIMO-MAC.
A separated source-channel coding scheme consists of concatenating a
quantizer of rate $R_s$ bits/source sample with a channel code of
rate $R_c$ bpcu, with $R_s = b R_c$. It was shown in \cite{HocZeg}
that the end-to-end distortion achievable by a separated scheme is
upper-bounded by
\[ D_{sep}(\rho) \leq D_Q(R_s) + \kappa P_e(\rho), \]
where $D_Q(R_s)$ denotes the quantizer distortion-rate function,
$P_e(\rho)$ is the error probability of the channel code, and
$\kappa$ is a constant independent of $\rho$. Let $R_s$ and $R_c$
denote the rates of the quantizer and the channel encoder
respectively, with corresponding multiplexing gains $r_s$ and $r_c$.
Thus,
\[ D_{sep}(\rho) \ \dot\leq \ \rho^{-b r_c} + \kappa \rho^{-d^*_{MAC}(r_c)}. \]
Notice that there exist very simple scalar quantizers with
rate-distortion optimal scaling $\rho^{-b r_c}$ \cite{CaiNar}, and
DMT optimal codes with finite blocklength $K N_t + M - 1$
\cite{TseVisZhe}. The best possible distortion SNR exponent is
obtained by choosing $r_c$ such that the two exponents of the above
expression are balanced, i.e., such that $b r_c = d^*_{MAC}(r_c)$. These considerations lead to:

\begin{thm} \label{th:separated}
If $N_t \leq \frac{M}{K+1}$, or if $N_t > \frac{M}{K+1}$ and $b >
\frac{K+1}{M} d^*_{N_t,M}\left( \frac{M}{K+1} \right)$, then the
distortion SNR exponent of the separated source channel coding
scheme
\begin{eqnarray*}
\delta_{sep}(b) \geq b\frac{j d^*_{N_t,M}(j-1) - (j-1)d^*_{N_t,M}(j)
} {b+d^*_{N_t,M}(j-1) - d^*_{N_t,M}(j)}, \\ b \in \left(
\frac{d^*_{N_t,M}(j)}{j}, \frac{d^*_{N_t,M}(j-1)}{j-1} \right],
\end{eqnarray*}
for $j = 1,\hdots,\min\{N_t,M\}$. Else, the exponent is given by
\begin{eqnarray*}
\delta_{sep}(b) \geq b \frac{j d^*_{KN_t,M}(j-1) - (j-1)
d^*_{KN_t,M}(j)}{b - K^2 [ d^*_{KN_t,M}(j) - d^*_{KN_t,M}(j-1) ]},
\end{eqnarray*}
for $b$ between $\frac{K d^*_{KN_t,M}(j)}{j}$ and $\frac{K
d^*_{KN_t,M}(j-1)}{j-1}$, for $j = 1,\hdots,\min \left\{
M,\frac{N_t}{K} \right\}$.
\end{thm}

\begin{proof}
For $N_t \leq \frac{M}{K+1}$, the MIMO-MAC DMT coincides with the
single-user DMT for all $r$. The optimal distortion exponent in this
case is the same as that for the single-user case, obtained in
\cite{CaiNar}. For the case when $N_t > \frac{M}{K+1}$, the
transition between the lightly loaded and the heavily loaded regimes
occurs at $r_c = \frac{M}{K+1}$, at which point the diversity gain
can be shown to be
\begin{eqnarray*}
d^*_{N_t,M}\left( \frac{M}{K+1} \right) & = & \left( \frac{M}{K+1} - \left\lfloor \frac{M}{K+1} \right\rfloor \right) 
\left[ 2 \left\lfloor \frac{M}{K+1} \right\rfloor  - (N_t + M - 1) \right] \\
& & + \left( N_t - \left\lfloor \frac{M}{K+1} \right\rfloor \right) \left( M - \left\lfloor \frac{M}{K+1} \right\rfloor \right).
\end{eqnarray*}
The slope of the line connecting the origin with $\left(
\frac{M}{K+1}, d^*_{N_t,M}\left( \frac{M}{K+1} \right) \right)$ is
$\frac{K+1}{M} d^*_{N_t,M}\left( \frac{M}{K+1} \right)$. Hence if $b
> \frac{K+1}{M} d^*_{N_t,M}\left( \frac{M}{K+1} \right)$, then once
again the single-user solution from \cite{CaiNar} holds. For the
case when $b < \frac{K+1}{M} d^*_{N_t,M}\left( \frac{M}{K+1}
\right)$, we need to solve
\begin{equation} \label{eq:equate_exponents}
b r_c = d^*_{KN_t,M}(Kr_c).
\end{equation}
For $r_c = \frac{j}{K}$, $j =
0,1,\hdots, \min \left\{ M, \frac{N_t}{K} \right\}$, we have
\[ b = \frac{K d^*_{KN_t,M}(j)}{j} \]
Then $\delta_{sep}(b)$ interpolates the points $\left( \frac{K
d^*_{K N_t, M}(j)}{j}, d^*_{K N_t, M}(j) \right)$ for $j =
0,1,\hdots,$$\min\left\{ M, \frac{N_t}{K} \right\}$.

For $r_c \in \left[ \frac{j-1}{K}, \frac{j}{K} \right]$, the
function $d^*_{K N_t, M}(Kr_c)$ linearly interpolates the points
$\left( \frac{j-1}{K}, d^*_{K N_t, M}(j-1) \right)$ and $\left(
\frac{j}{K}, d^*_{K N_t, M}(j) \right)$. This gives us that
\begin{eqnarray*}
d^*_{K N_t, M}(Kr_c) & = & K \left( d^*_{K N_t, M}(j) - d^*_{K N_t,
M}(j-1) \right) \left[ Kr_c - \frac{j-1}{K} \right] + d^*_{K N_t, M}(j-1), \\
& & \;\;\; \mbox{for} \;\;\;  r_c \in \left[ \frac{j-1}{K}, \frac{j}{K} \right].
\end{eqnarray*}
Solving \eqref{eq:equate_exponents}, we hence obtain the exponent
\[ b r_c = b \frac{j d^*_{KN_t,M}(j-1) - (j-1) d^*_{KN_t,M}(j)}{b - K^2 \left[ d^*_{KN_t,M}(j) - d^*_{KN_t,M}(j-1) \right]}, \]
for $b$ between $K \frac{d^*_{KN_t,M}(j)}{j}$ and $K
\frac{d^*_{KN_t,M}(j-1)}{j-1}$.

\end{proof}

For the special case when $N_t=1$ and $M=K$, $d^*_{MAC}(r_c) =
K(1-r_c)$, leading us to choose
\begin{equation}\label{eq:delta_opt_separated}
r_c = \frac{K}{K+b}  \ \ \Rightarrow \ \  \delta_{sep}(b) \geq
\frac{bK}{K+b}.
\end{equation}

\subsection{Performance of successive interference cancellation (SIC) receivers\label{sec:SIC}}
We now consider the performance of low complexity MMSE and ZF SIC
receivers (i.e., MMSE and ZF V-BLAST receivers \cite{WolFosGolVal})
in terms of distortion SNR exponents. We show that these low
complexity receivers are very suboptimal (irrespective of the order
in which users are decoded), which suggests that joint decoding
techniques might be inevitable in this scenario. The suboptimality
in terms of distortion exponent is a direct consequence of the
suboptimality of both MMSE and ZF SIC receivers in terms of DMT. For
simplicity of exposition we restrict ourself to the case of $N_t =
1$ and $M = K$. It was shown in \cite{TseVisZhe} that these
receivers with no user ordering achieve a DMT of $d^*_{SIC}(r) =
1-r_c$. It was later shown in \cite{JiaVarLi} that no ordering,
including the V-BLAST optimal ordering \cite{WolFosGolVal} can
improve upon this DMT. In the separated case, this would result in a
very suboptimal distortion-SNR exponent of
\[ \delta_{sep,SIC}(b) = \frac{b}{1+b}. \]
Since $\delta(b) \geq 1$ is needed in order to achieve bounded or vanishing rate gap for the MIMO downlink (see Sec.~\ref{sec:channel_model}), we conclude that a digital ``separated'' scheme based on SIC fails to achieve this goal, and would result inevitably in a loss of degrees of freedom in the MIMO downlink due to the poor SNR exponent of the CSIT feedback.

\section{Hybrid Digital-Analog Coding Scheme for the MIMO-MAC\label{sec:hybrid}}

While the separated scheme is close to optimal for very low and very
high bandwidth efficiencies, it is in fact very suboptimal for a
wide range of bandwidth efficiencies, including the ones that are
most relevant for practical CSIT feedback scheme, such as $b$
between $2$ and $6$. One therefore needs to consider joint source
channel coding to improve upon the performance of the separated
source channel coding scheme. This approach has been employed with
good success in the case of the single-user MIMO channel in the
series of works in \cite{CaiNar,GunErk,BhaNarCai}. In the following
subsection, we consider a generalization of the hybrid
digital-analog scheme of \cite{CaiNar} for the MIMO-MAC, and
evaluate the distortion SNR exponent for this case. 

We restrict our attention to the case of bandwidth expansion, i.e.,
$b \geq 1$ (the case of bandwidth compression $b<1$ is not relevant
for our purpose since it would result in a loss of downlink
multiplexing gain \cite{CaiJinKobRav}). Along the lines of
\cite{CaiNar}, we consider the encoder shown in
Fig.~\ref{fig:BW_expansion} for each UT.

\begin{figure}[h]
\begin{center}
\centerline{\includegraphics[width=12cm]{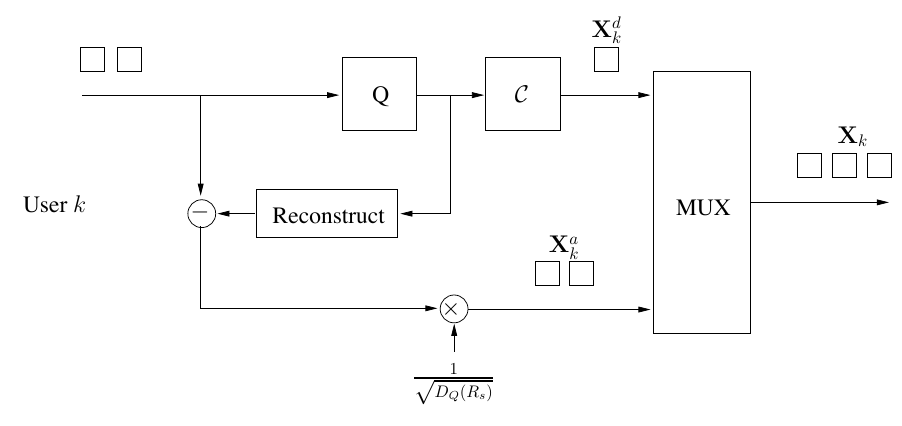}}
\caption{Hybrid digital-analog scheme for bandwidth
expansion\label{fig:BW_expansion}}
\end{center}
\end{figure}

We assume that $M \geq K$ in this section. Define $N_t'$ to be the
maximum integer less than or equal to $N_t$ such that $K N_t' \leq
M$. The block of $S$ input symbols corresponding to the
$i^{\text{th}}$ UT is first fed to a separated source-channel
encoder that produces ``digital'' encoded data $\Xm_i^d$ of size
$N_t \times T_d$. The quantization error $\ev$ is scaled by
$1/\sqrt{N_t' D_Q(R_s)}$ and reformatted to produce ``analog'' data
$\Xm_i^a$ of size $N_t' \times S/N_t'$. The analog data is to be
transmitted over a pre-selected set of $N_t'$ antennas. The digital
and analog data are multiplexed in time to produce the signal
$\Xm_i$ transmitted by the $i^\text{th}$ UT, over $T = T_d + S/N_t'$
time-slots. We have that
\[ r_c = \frac{S}{T_d} r_s \ \ \text{ and } \ \ b = \frac{T_d+S/N_t'}{S},\]
which result in
\[ r_c = \frac{N_t' r_s}{N_t' b-1}.\]

The BS receives a faded superposition of the signals transmitted by
the various UT in AWGN, according to \eqref{eq:MAC}. The BS first
tries to decode the digital data, by using the first $T_d$ columns
of $\Ym$, denoted as $\Ym^{d}$. Up to SNR exponents, the probability
of error can be replaced with the probability of outage for the MAC.
The outage event is given by $\Oc = \left\{ \{ \Hm_i \}_{i=1}^{K}:
(R_c,\hdots,R_c) \notin \Cc_{\text{MAC}}(\Hm; \rho) \right\}$.

If the decoding is successful, the BS obtains a linear MMSE estimate
of the estimation error $\ev$ from the last $S/N_t'$ columns of
$\Ym$ (denoted as $\Ym^a$), and adds it to the decoded digital
information to produce a reconstruction. The following theorem presents the distortion SNR exponent for this hybrid digital-analog coding scheme (while an explicit solution may be obtained along the lines of Theorem~\ref{th:separated}, we refrain from the details in the general case for reasons of brevity, and present the closed-form solution only for the case of $N_t = 1$, $M=K$).

\begin{thm} \label{th:hybrid}
The distortion SNR exponent for the hybrid digital-analog coding scheme for the MAC is given by
\[ \delta_{hyb}(b) \geq d^*_{MAC}(r^*), \]
where $r^*$ is the value of $r_c$ that solves
\begin{eqnarray*}
&1 + \frac{r_c}{N_t'}(N_t' b-1) = d^*_{MAC}(r_c)&,
\end{eqnarray*}
and $N_t'$ is the maximum integer less than or equal to $N_t$ such that $K N_t' \leq M$.
In particular, when $N_t = 1$ and $M=K$ the distortion SNR exponent is given by
\[ \delta_{hybrid}(b) \geq 1 + \frac{(b-1)(K-1)}{K+b-1}. \]
\end{thm}

\begin{proof}
Let $\Hm_k^a$
denote the $M \times N_t'$ channel from the $k^{\text{th}}$ UT to
the BS corresponding to the $N_t'$ antennas chosen for analog
transmission and $\Hm^a \triangleq [\Hm_1^a \ \cdots \ \Hm_K^a]$.
Notice that
\begin{eqnarray*}
\Ym^a &=& \sqrt{\rho} \sum_{k=1}^{K} \Hm_k^a \Xm_k + \Wm^a\\
&=& \sqrt{\rho} \underbrace{[\Hm_1^a \Hm_2^a \cdots
\Hm_K^a]}_{\triangleq \Hm^a} \underbrace{\left[
\begin{array}{c}
\Xm^a_1\\
\Xm^a_2\\
\vdots\\
\Xm^a_K
\end{array}
\right]}_{\triangleq \Xm^a} + \Wm^a.
\end{eqnarray*}
Stacking the columns of $\Ym^a$ and $\Xm^a$ one below the other to
obtain $\yv^a$ and $\xv^a$ respectively, we obtain the equivalent
vector model
\[ \yv^a = \sqrt{\rho} \underbrace{\left[
\begin{array}{cccc}
  \Hm^a &  &  &  \\
    & \Hm^a &  &  \\
   &  & \ddots &  \\
   &  &  & \Hm^a
\end{array}
\right]}_{\triangleq \Hm_{eq} = \Id_{S/N_t'} \otimes \Hm^a} \xv^a +
\wv^a\]

Notice that we may obtain an LMMSE estimate of $\ev$ by first
linearly estimating $\xv^a$ from $\yv^a$, and then by rearranging
components and scaling. In order to estimate $\xv^a$ from $\yv^a$,
we compute the covariance matrices
\begin{eqnarray*}
K_{\xv^a {\xv^a}^*} &=& \Id\\
K_{\xv^a {\yv^a}^*} &=& \sqrt{\rho} \Hm_{eq}^\herm\\
K_{\yv^a {\yv^a}^*} &=& \Id + \rho \Hm_{eq} \Hm_{eq}^\herm.
\end{eqnarray*}
From the orthogonality principle, the minimum mean squared error
\begin{eqnarray*}
\mmse(\xv^a|\Hm) &\doteq& \trace \left( \Id - \rho \Hm_{eq}^\herm
(\Id +
\rho \Hm_{eq} \Hm_{eq}^\herm)^{-1} \Hm_{eq} \right)\\
&=& \trace \left[ (\Id + \rho \Hm_{eq}^\herm \Hm_{eq})^{-1}
\right]\\
&=& \sum_{k=1}^{m'} \frac{S/N_t'}{1 + \rho \lambda_k},
\end{eqnarray*}
where $m' \triangleq \min\{ M, K N_t'\}$ and $\lambda_1 \leq \cdots
\leq \lambda_{m'}$ denote the ordered eigenvalues of $\Hm^\herm
\Hm$. Thus
\begin{eqnarray*}
\mmse(\ev|\Hm) &=& D_Q(R_s) \sum_{k=1}^{m'} \frac{S/N_t'}{1 + \rho
\lambda_k}.
\end{eqnarray*}
It can be shown that the distortion for the hybrid scheme is given
by
\begin{eqnarray} \label{eq:D_hybrid}
D_{hybrid}(\rho) &\dot\leq& \int_{\Oc^c} \mmse(\ev|\Hm) P(\Hm) d\Hm
+
\kappa P(\Oc) \nonumber \\
&\leq& \int \mmse(\ev|\Hm) P(\Hm) d\Hm + \kappa P(\Oc)
\end{eqnarray}
where $\kappa$ is a constant independent of $\rho$.

Let $\lambda_i \doteq \rho^{-\alpha_i}$. Using the joint pdf
$p(\alphav)$ given in \eqref{eq:joint_pdf}, we compute
\begin{eqnarray*}
\int \mmse(\ev|\Hm) P(\Hm) d\Hm &\doteq&
\rho^{-r_s} \int\limits_{\alpha_i \geq 0}
\rho^{-(1-\alpha_1)^+}  \rho^{- \sum_{i=1}^{m'} (2i-1+|M-KN_t'|)\alpha_i} d\alphav \\
&\doteq& \rho^{-r_s} \rho^{- \inf\limits_{\alpha_i \geq 0 \ \forall
\ i} f(\alphav) },
\end{eqnarray*}
where $f(\alphav) = (1-\alpha_1)^+ + \sum_{i=1}^{m'} (2i-1+|M-KN_t'|)\alpha_i$. Hence
\begin{equation*}
\int \mmse(\ev|\Hm) P(\Hm) d\Hm \doteq \rho^{-(1+r_s)}.
\end{equation*}
The distortion SNR exponent for the hybrid
scheme is obtained by equating the SNR exponents of the two terms in \eqref{eq:D_hybrid}. 
\end{proof}

Similar to the separated source channel coding case, it can be shown
that successive interference cancellation is very suboptimal in this
case also. Fig.~\ref{fig:distor_exponents_K4} shows plots of the
distortion exponent for $K = 4$ obtained from separated source
channel coding and hybrid digital-analog coding in comparison to the
informed transmitter upper bound and the analog scheme.
\begin{figure}
\centering 
\includegraphics[width=10cm]{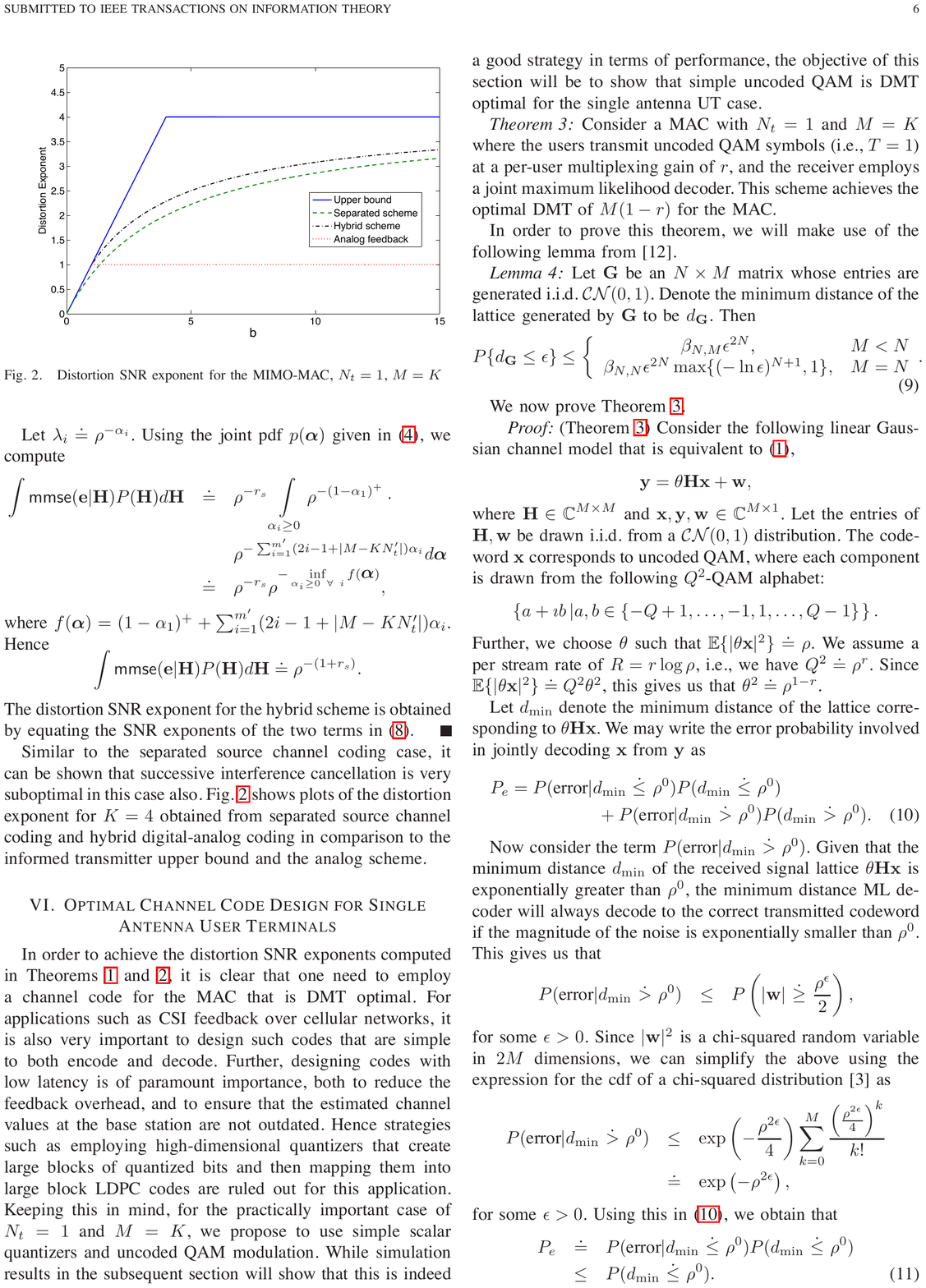}
\caption{Distortion SNR exponent for the MIMO-MAC, $N_t = 1$, $M =
K$\label{fig:distor_exponents_K4}}
\end{figure}

\section{Optimal Channel Code Design for Single Antenna User Terminals\label{sec:QAM_DMT}}

In order to achieve the distortion SNR exponents computed in Theorems~\ref{th:separated} and \ref{th:hybrid}, 
a DMT optimal code for the MIMO MAC is needed. 
For applications such as CSIT feedback over cellular networks, it is very important to design such codes that are simple to both encode and decode. Further, designing codes with low latency is of paramount importance, both to reduce the feedback overhead, and to ensure that the estimated channel values at the base station are not outdated. 
Hence strategies such as employing high-dimensional quantizers that create large blocks of quantized bits and then mapping them into 
channel codes with large block length (e.g., LDPC codes) 
are ruled out for this application. Keeping this in mind, for the practically important case of $N_t = 1$ and $M = K$, we propose to use 
simple scalar quantization and uncoded QAM modulation. 
While simulation results in the subsequent section will show that this is indeed a very good strategy in terms of performance, 
the objective of this section will be to show that uncoded QAM is indeed DMT optimal for the MIMO-MAC with single 
antenna UTs.

\vspace{12pt}

\begin{thm} \label{th:QAM_DMT}
Consider a MAC with $N_t = 1$ and $M = K$ where the users transmit uncoded QAM symbols (i.e., $T = 1$) at a per-user multiplexing gain of $r$, and the receiver employs a joint maximum likelihood decoder. This scheme achieves the optimal DMT of $M(1-r)$ for the MAC.
\end{thm}

\vspace{12pt}

\noindent
In order to prove this theorem, we will make use of the following lemma from \cite{TahMobKha}.

\vspace{12pt}

\begin{lem} \label{lem:khandani}
Let $\Gm$ be an $N \times M$ matrix whose entries are generated i.i.d. $\Cc \Nc (0,1)$. Denote the minimum distance of the lattice generated by $\Gm$ to be $d_{\Gm}$. Then
\begin{equation} \label{eq:lemma}
P \{ d_\Gm \leq \epsilon \} \leq \left\{
\begin{array}{cc}
\beta_{N,M} \epsilon^{2N}, & M<N\\
\beta_{N,N} \epsilon^{2N} \max\{ (- \ln \epsilon)^{N+1}, 1 \}, & M=N
\end{array}
. \right. 
\end{equation}
\end{lem}

We now prove Theorem~\ref{th:QAM_DMT}. 

\begin{proof} (Theorem~\ref{th:QAM_DMT})
Consider the following linear Gaussian channel model that is equivalent to \eqref{eq:MAC},
\[ \yv = \theta \Hm \xv + \wv, \]
where $\Hm \in \CC^{M \times M}$ and $\xv, \yv, \wv \in \CC^{M \times 1}$. Let the entries of  $\Hm, \wv$  be drawn i.i.d. from a $\Cc \Nc (0,1)$ distribution.  The codeword $\xv$ corresponds to uncoded QAM, where each component is drawn from the following $Q^2$-QAM alphabet:
\[ \left\{ a + \imath b \left|  a,b \in \{ -Q+1, \hdots, -1,1,\hdots, Q-1 \} \right. \right\}. \]
Further, we choose $\theta$ such that $\EE \{ |\theta \xv|^2 \} \doteq \rho$. We assume a per stream rate of $R = r \log \rho$, i.e., we have $Q^2 \doteq \rho^r$. Since $\EE \{ |\theta \xv|^2 \} \doteq Q^2 \theta^2$, this gives us that $\theta^2 \doteq \rho^{1-r}$.

Let $d_{\min}$ denote the minimum distance of the lattice corresponding to $\theta \Hm \xv$. We may write the error probability involved in jointly decoding $\xv$ from $\yv$ as
\begin{equation} \label{eq:error_prob}
P_e =  P(\text{error}| d_{\min} \ \dot\leq \ \rho^0) P(d_{\min} \ \dot\leq \ \rho^0) + 
P(\text{error}| d_{\min} \ \dot> \ \rho^0) P(d_{\min} \ \dot> \ \rho^0).
\end{equation}

Now consider the term $P(\text{error}| d_{\min} \ \dot> \ \rho^0)$. Given that the minimum distance $d_{\min}$ of the received signal lattice $\theta \Hm \xv$ is exponentially greater than $\rho^0$, the minimum distance ML decoder will always decode to the correct transmitted codeword if the magnitude of the noise is exponentially smaller than $\rho^0$. This gives us that
%
\begin{eqnarray*}
P(\text{error}| d_{\min} \ \dot> \ \rho^0) &\leq& P \left(|\wv| \ \dot\geq \ \frac{\rho^\epsilon}{2}\right),
\end{eqnarray*}
for some $\epsilon > 0$. Since $|\wv|^2$ is a chi-squared random variable in $2M$ dimensions, we can simplify the above using the expression for the cdf of a chi-squared distribution \cite{Pro} as
\begin{eqnarray*}
P(\text{error}| d_{\min} \ \dot> \ \rho^0) &\leq& \exp \left( -\frac{\rho^{2\epsilon}}{4} \right) \sum_{k=0}^{M} \frac{\left( \frac{\rho^{2\epsilon}}{4} \right)^k}{k!}\\
&\doteq& \exp \left( - \rho^{2 \epsilon} \right),
\end{eqnarray*}
for some $\epsilon > 0$. Using this in \eqref{eq:error_prob}, we obtain that
\begin{eqnarray} \label{eq:Pe_simplify}
P_e &\doteq& P(\text{error}| d_{\min} \ \dot\leq \ \rho^0) P(d_{\min} \ \dot\leq \ \rho^0) \nonumber\\
&\leq& P(d_{\min} \ \dot\leq \ \rho^0).
\end{eqnarray}

Noticing that $d_{\min} = \theta d_\Hm$, we obtain from \eqref{eq:Pe_simplify} and an application of Lemma~\ref{lem:khandani} that
\begin{eqnarray*}
P_e &\dot\leq& P \left(d_\Hm \ \dot\leq \ \frac{1}{\rho^{(1-r)/2}} \right)\\
&\doteq& \frac{1}{\rho^{M(1-r)}} (\ln \rho)^{n+1}\\
&\doteq& \rho^{-M(1-r)}
\end{eqnarray*}
\end{proof}

\begin{note}
A closer look at the proof of Theorem~\ref{th:QAM_DMT} will reveal that the same proof will carry over unchanged to the case where we use a minimum distance lattice decoder. By a minimum distance lattice decoder, we mean a sphere decoder \cite{DamGamCai} that decodes to the whole infinite lattice and not just the finite constellation carved from the lattice, without doing any special precoding or weighting. Such a decoder is referred to as a ``naive lattice decoder'' in \cite{GamCaiDam_LAST}, in contrast to the MMSE-GDFE lattice decoder \cite{GamCaiDam_LAST}, or the regularized lattice decoder \cite{JalEli}.
\end{note}

\section{Simulation results\label{sec:simul}}
In this section, we present simulation results comparing the case of
digital and analog CSIT feedback over the MIMO-MAC. We focus on the
case of $N_t=1$ and $K=M=4$ and restrict attention to separated
source channel coding for the digital feedback case. Further, we focus on the following extremely simple encoder, comprising of a uniform scalar quantizer \cite{CaiNar} followed by uncoded QAM for the channel code (these are respectively optimal with respect to quantizer distortion, and DMT). We will restrict attention to the case of joint ML decoding at the receiver (implemented using a sphere decoder \cite{DamGamCai}).


\begin{figure}
\centering
\includegraphics[width=10cm]{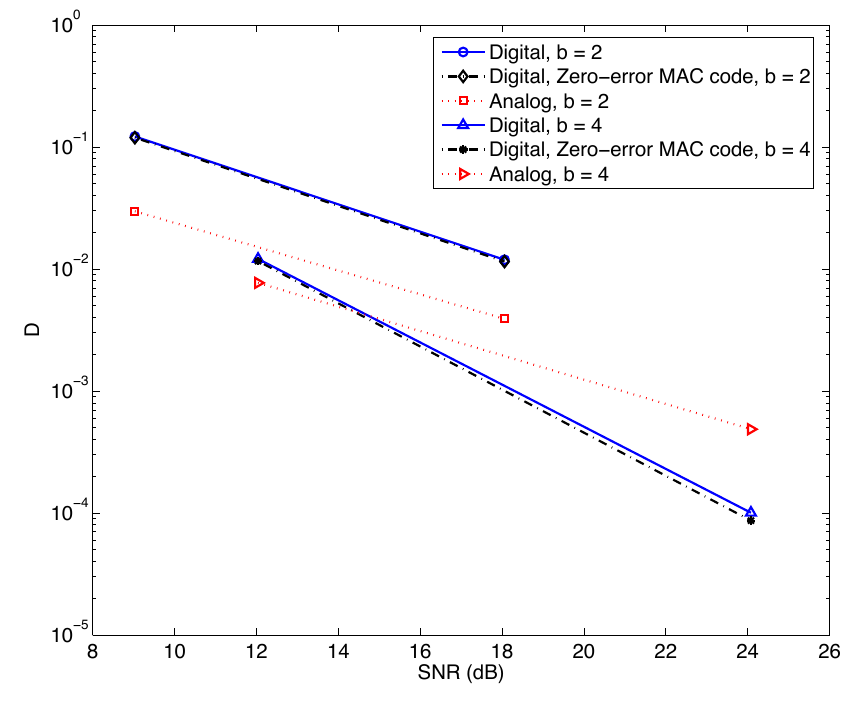}
\caption{Comparing analog feedback with digital feedback using
separated source-channel coding, $K = M = 4$, $N_t =
1$\label{fig:distor_K4_b2b4}}
\end{figure}
Fig.~\ref{fig:distor_K4_b2b4} compares the per-user squared error
distortion using analog and digital feedback for the cases of $b=2$
and $b=4$. For the analog feedback case, we use time-division with 2
users transmitting un-quantized Gaussian sources simultaneously
while satisfying the average power constraint (such a time-division
strategy has been proven to be the optimal strategy for analog
feedback \cite{CaiJinKobRav}). Notice that we fix $T = 4$ for the
digital case, and vary the number of samples of the transmitted
Gaussian source $S = 2,1$ to achieve $b=2,4$ respectively. We vary
the codebook size with SNR according to the optimal multiplexing
gains $r_c = 2/3$ and $r_c = 1/2$ corresponding to $b=2,4$
respectively, computed using \eqref{eq:delta_opt_separated}. For
example, at $b=4$, this would correspond to $4$-QAM and $16$-QAM
constellations at SNRs of $12.0412$ dB and $24.0824$ dB respectively
(in turn corresponding to rates $R_c = 2,4$ bpcu and $R_s = 8,16$
bits/complex sample respectively). We make use of scalar uniform
quantizers from \cite{CaiNar} and use gray maps to assign bits to
the QAM alphabets. While the slope of analog feedback can never
exceed 1, we see that the slopes of the digital feedback in
Fig~\ref{fig:distor_K4_b2b4} are close to their theoretical values
of $\delta_{sep}(2) \geq \frac{4}{3}$ and $\delta_{sep}(4) \geq 2$.
While analog feedback is better at $b=2$, we observe that the
digital feedback significantly outperforms analog feedback at $b=4$.
In order to illustrate that the channel code in this example 
is indeed near optimal, we also plot the distortion achieved by an ideal
zero-error MAC code, for which the only distortion is due to the source quantization. 
We observe that the distortion performance of the chosen MAC code is almost identical to the ideal
zero-error case. This does not mean that uncoded QAM achieves zero error at these values of SNRs, of course.
It just indicates that the probability of error scales with SNR together with the quantization distortion, such that 
the performance is effectively given by the quantization distortion, in agreement with the 
optimization of the SNR exponent in Sections \ref{sec:separated} and \ref{sec:hybrid}.

In a practical cellular system, the SNR experienced by UTs close to
the BS and those at the cell fringes may differ by several tens 
of dBs, owing to significantly different path-loss. As we saw before,
the downlink rate-gap is determined by $\rho D(\rho)$, which
necessitates that the distortion decays with an exponent strictly
larger than 1 to ensure a vanishing rate gap. 
Analog feedback is insufficient for this purpose, leading to a rate-gap that
does not decay with SNR (the presence of sub-logarithmic terms
results in the exponent being slightly less than one in practice).
Digital feedback is hence a necessity in this practical scenario.
Corresponding to this scenario, we present simulations of a $K=4$
user system, where two high-SNR users operate at 24 dB, and two
low-SNR users operate at 12 dB. We fix $b = 4$, and choose the rates
of the users in accordance with the distortion-optimal multiplexing
gain of $r_c = 1/2$ (hence low SNR users signal using $4$-QAM and
high SNR users with $16$-QAM).
\begin{figure}
\centering
\includegraphics[width=10cm]{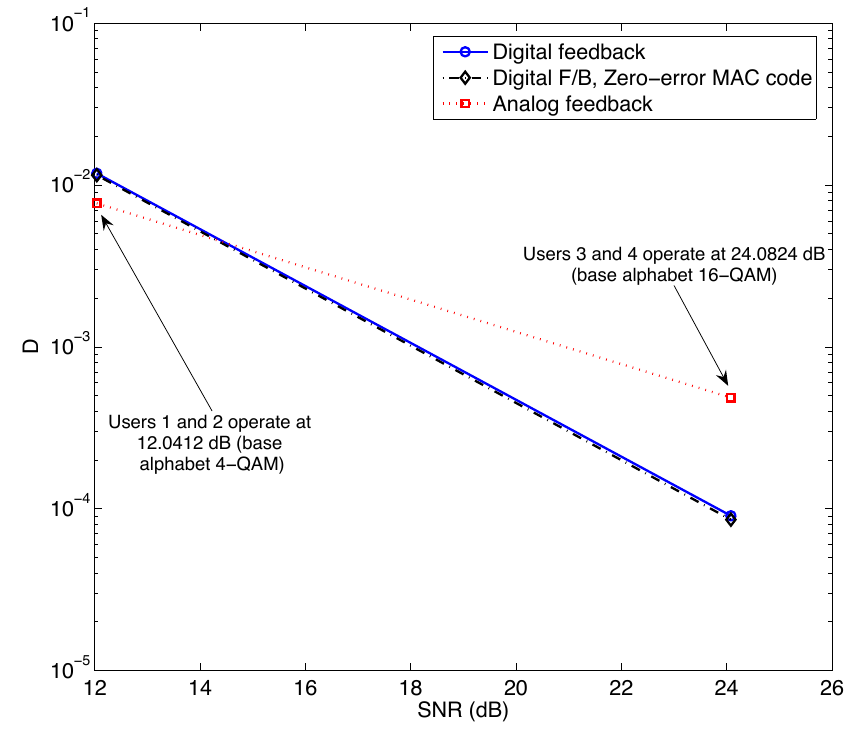}
\caption{Distortion obtained for users with very different SNRs, $K
= M = 4$, $N_t = 1$\label{fig:distor_K4_b4_assymetric}}
\end{figure}
We plot the per-user distortions obtained for the low and high SNR
users against their respective SNRs in
Fig.~\ref{fig:distor_K4_b4_assymetric}, and observe that the
distortion decays with a slope that is exactly equal to the slope
for the case of $b=4$ in Fig.~\ref{fig:distor_K4_b2b4}. This
provides a very strong motivation for the use of intelligently
designed digital CSIT feedback schemes in practice, and is also testament
to the practical relevance of the theoretical distortion exponents derived in this paper.

\section{Conclusions}

In this paper, we considered the very important practical problem of designing CSIT 
feedback schemes for cellular systems. We modeled this problem as the transmission of Gaussian sources over a MIMO-MAC, and argued that an appropriate metric to consider in order to maximize the downlink rate in the cellular system is the mean squared error distortion at which the channel is estimated at the base station. An upper bound was obtained on the asymptotic SNR-exponent of the distortion, and some 
CSIT feedback strategies were compared in terms of their achievable distortion SNR exponent. While traditionally, system designers have considered using analog (unquantized) feedback, we showed that such a strategy can at best lead only to an achievable downlink rate that is at a 
constant gap from the downlink rate achievable with perfect CSIT. We designed two ``digital'' feedback schemes, using separated and joint source channel coding respectively, that lead to an achievable downlink rate that approaches the optimal CSI downlink rate as the SNR increases (i.e., the rate gap vanishes).

We then turned our attention to another very important system design problem: the design of explicit quantization and channel coding 
schemes for the proposed digital feedback scheme. 
Simplicity of such schemes is very important for the application of CSIT feedback: complicated structures employing long 
outer codes and high-dimensional vector quantizers would increase the feedback overhead considerably, and the 
resulting high latency (large blocklengths for the feedback scheme) would render the CSIT fed back completely outdated and hence useless
for the purpose of multiuser MIMO downlink precoding. We showed that a very simple separated scheme comprising of a uniform scalar quantizer and uncoded QAM modulation is optimal in terms of distortion SNR exponent for the case of single antenna user terminals. Further, we demonstrated through simulations that such a scheme outperforms analog feedback in practical scenarios. 
Our simulations also reveal the advantage of using variable rate feedback for users with different operating SNRs 
in a cellular system with large pathloss dynamic range. 
A naive feedback scheme with a fixed number of bits would result in an interference-limited system,  while our digital feedback scheme  results in significant performance gains, especially for the users in very good SNR conditions.


%

\end{document}